\documentclass[11pt]{article}

\usepackage[margin=1in]{geometry}
\usepackage{amsmath,amssymb,amsthm}
\usepackage{mathtools}
\usepackage[protrusion=true,expansion=false]{microtype}
\usepackage[hidelinks]{hyperref}
\usepackage{enumitem}
\usepackage{graphicx}
\usepackage{parskip}

\newcommand{\E}{\mathbb{E}}
\newcommand{\Q}{\mathbb{Q}}
\newcommand{\Prob}{\mathbb{P}}
\newcommand{\F}{\mathcal{F}}

\theoremstyle{plain}
\newtheorem{proposition}{Proposition}
\newtheorem{corollary}{Corollary}

\theoremstyle{remark}

\title{\vspace{-2em}Perpetual Futures for Stocks:\\
The SpaceX Pre-IPO Market}
\author{
Aditya Gupta\\
\small Stochastic Process, New York
\and
Nicholas G. Polson\\
\small Booth School of Business, University of Chicago
}
\date{\today}

\begin{document}
\maketitle
\vspace{-1.0em}

\begin{abstract}
\noindent
Robert Shiller proposed perpetual futures in 1993 to create derivative markets
for assets that are illiquid or whose price cannot be observed directly, such as
single family homes, human capital, and the consumer price index. The crypto
markets later built the instrument for a different reason and with a different
funding rule. We give a single no arbitrage result that nests both designs: the
perpetual price is the present value of a benchmark flow discounted at the
funding rate, so the funding rule chooses both the benchmark and the discount.
We give a random time change representation in which the price is the expected
spot at the first event of a clock whose intensity is the funding rate, use it to
show that stochastic volatility moves the basis only through the carry, so a
volatility risk premium and not volatility itself can break the peg, read price
discovery as the
convergence of a Doob martingale driven by a stochastic approximation, and give a
segmented market equilibrium that makes the pre listing premium structural rather
than behavioural. The June 2026 SpaceX pre-IPO market is the first large scale
realization of the idea for an equity claim, and we find that the perpetual
consensus forecast the secondary clearing price more accurately than the
bookbuilt offer. We close with other equity applications.
\end{abstract}

\medskip
\noindent\textbf{Keywords:} perpetual futures; funding rate; price discovery; no
arbitrage; martingale pricing; nonlinear filtering; stochastic volatility; time
change; market segmentation; pre-IPO equity; generative Bayesian computation.


\section{Introduction}\label{sec:intro}

Derivative markets usually presuppose a liquid, observable spot against which a
contract settles. Many economically important claims have no such spot: a single
family home, a unit of human capital, a basket of macro risk, or the equity of a
private firm. Shiller (1993a) proposed the perpetual future to reach these
claims, settling a non expiring contract each period against an observable flow
rather than a price. The design lay mostly dormant until cryptocurrency venues
rebuilt it with a different settlement, a funding payment keyed to the basis
against a spot, and made it the dominant traded instrument in that market. In
June 2026 the same instrument was used at equity scale for the first time, to
trade SpaceX before and through its public listing.

This paper gives a unified treatment and applies it to that episode. First, a
single no arbitrage result, Proposition~\ref{prop:pv}, shows that any perpetual
prices as the present value of a benchmark flow discounted at the funding rate,
so the funding rule alone separates the Shiller design from the crypto design and
the two are one instrument with two benchmarks. Second, Proposition~\ref{prop:sub}
gives a random time change representation: the price is the expected spot at the
first event of a clock whose intensity is the funding rate. Third,
Proposition~\ref{prop:exact} shows that stochastic volatility leaves the peg
intact unless it is priced into the carry, the perpetual analogue of the futures
minus forward adjustment. Fourth, we read price discovery as filtering and show
in Section~\ref{sec:bayes} that the funding rule is a feedback observer, the
price a Doob martingale of the latent value, and the peg a Bayesian fixed point.
Fifth, Proposition~\ref{prop:seg} gives a segmented market equilibrium for the
basis. Sixth, we work the SpaceX example and find that the perpetual consensus
forecast the secondary clearing price better than the offer. We write under a
fixed risk neutral measure $\Q$ throughout the pricing sections and pass to the
physical measure $\Prob$ only when reading the premium.

The paper draws three literatures together. The first is Shiller (1993a, 1993b)
on cash settled markets for unobservable assets, of which the perpetual future is
one instrument. The second is the recent no arbitrage analysis of crypto
perpetuals of He, Manela, Ross and von Wachter (2022) and Ackerer, Hugonnier and
Jermann (2024), which we nest and extend to a latent underlying. The third is the
filtering and martingale theory of price formation, from Kalman (1960) and the
martingale pricing of Harrison and Kreps (1979) to generative Bayesian
computation (Polson and Sokolov 2024), which we use to read the funding rule as a
learning rule. The contribution is to show that these are one picture: the funding
rate is at once a discount rate, a clock, an observer gain, and a learning rate.

\section{Shiller's construction}\label{sec:shiller}

Shiller's idea is to settle a contract every period against a flow index rather
than against a transactable spot price. Let $f_t$ be the perpetual price at day
$t$, let $D_t$ be the dividend or rent index for day $t$, and let $r_t$ be the
one period return on a low risk alternative asset over $[t,t+1]$. The daily cash
settlement paid from short to long is
\begin{equation}
s_t = D_t - r_t\,f_t .
\end{equation}
The long receives the index flow and pays financing on the contract price. The
position is opened at zero cost, so over $[t,t+1]$ the long earns the capital
gain plus the settlement,
\[
\Pi_t = (f_{t+1}-f_t) + s_t = f_{t+1} - f_t + D_t - r_t f_t .
\]
Imposing the no arbitrage condition $\E_t^{\Q}[\Pi_t]=0$ gives the pricing
recursion
\begin{equation}
f_t = \frac{\E_t^{\Q}[f_{t+1}] + D_t}{1+r_t},
\end{equation}
which iterates forward, under a transversality condition, to
\begin{equation}\label{eq:shiller}
f_t = \sum_{k=0}^{\infty}
\frac{\E_t^{\Q}\!\left[D_{t+k}\right]}{\prod_{j=0}^{k}(1+r_{t+j})}.
\end{equation}
The perpetual price equals the present value of the expected flow. With $r_t=r$
and $\E[D]=D$ constant this collapses to the Gordon form $f=D/r$. The point is
structural: the contract recovers a fundamental value using only an observable
flow, so it prices assets that are never traded and whose spot is never seen.
That is why Shiller aimed it at real estate, labour income, and price indices.

\section{A unified perpetual: present value at the funding rate}\label{sec:unified}

Work in continuous time with a general settlement rule.
Fix a risk neutral measure $\Q$ and a filtration $(\F_t)$. Consider a perpetual
with price $F_t$ whose holder receives an instantaneous settlement
\begin{equation}
d\Sigma_t = (a_t - b_t F_t)\,dt,
\end{equation}
where $a_t$ and $b_t>0$ are adapted, $a$ is the benchmark flow and $b$ is the
funding intensity. The two existing designs are the two choices of $(a,b)$ given
below. Entry is costless, so the cumulative gain $G_t = F_t + \Sigma_t$ is a
$\Q$ martingale.

\begin{proposition}[Perpetual price as a present value at the funding rate]
\label{prop:pv}
Let $G_t=F_t+\Sigma_t$ be a $\Q$ martingale with settlement
$d\Sigma_t=(a_t-b_t F_t)\,dt$, $b_t>0$, $\int_0^\infty b_u\,du=\infty$, and
suppose the transversality condition
$\lim_{T\to\infty}\E_t^{\Q}\!\left[e^{-\int_t^T b_u du}F_T\right]=0$ holds. Then
\begin{equation}\label{eq:pv}
F_t=\E_t^{\Q}\!\left[\int_t^{\infty}
e^{-\int_t^{s} b_u\,du}\,a_s\,ds\right].
\end{equation}
\end{proposition}

\begin{proof}
The martingale property gives $\E_t[dF_t]=-\E_t[d\Sigma_t]=(b_tF_t-a_t)\,dt$.
Write $\Lambda_t=\exp(-\int_0^t b_u\,du)$, so
$d(\Lambda_s F_s)=\Lambda_s(dF_s-b_sF_s\,ds)$ and hence
$\E_s[d(\Lambda_s F_s)]=\Lambda_s(\E_s[dF_s]-b_sF_s\,ds)=-\Lambda_s a_s\,ds$.
Integrating from $t$ to $T$ and taking $\E_t$,
\[
\E_t[\Lambda_T F_T]-\Lambda_t F_t=-\,\E_t\!\int_t^T \Lambda_s a_s\,ds .
\]
Divide by $\Lambda_t$, let $T\to\infty$, and apply transversality.
\end{proof}

The perpetual price is the present value of the benchmark flow $a$ discounted at
the funding rate $b$. The funding rule fixes both the flow being averaged and the
rate at which the average decays.

\begin{corollary}[The two designs]
\label{cor:cases}
\emph{(i) Shiller.} With $a_t=D_t$ and $b_t=r_t$,
$F_t=\E_t^{\Q}\!\int_t^\infty e^{-\int_t^s r_u du}D_s\,ds$, the present value of
the dividend or rent flow at the riskless rate, the continuous time form of
\eqref{eq:shiller}. \;
\emph{(ii) Crypto basis peg.} With observable spot $S_t$ and funding paid by the
long $g_t=\kappa\,(F_t-S_t)$, the settlement received by the long is
$-\kappa(F_t-S_t)\,dt$, so $a_t=\kappa S_t$ and $b_t=\kappa$, giving
\begin{equation}\label{eq:ewma}
F_t=\kappa\,\E_t^{\Q}\!\left[\int_t^\infty e^{-\kappa(s-t)}S_s\,ds\right],
\end{equation}
an exponentially weighted average of the expected future spot with horizon
$1/\kappa$.
\end{corollary}

\begin{corollary}[The peg and its fast funding limit]
\label{cor:peg}
If $S$ is a $\Q$ martingale then $F_t=S_t$ for every $\kappa>0$. More generally,
if $s\mapsto\E_t^{\Q}[S_s]$ is right continuous at $s=t$, then $F_t\to S_t$ as
$\kappa\to\infty$.
\end{corollary}

\begin{proof}
The kernel $\kappa e^{-\kappa(s-t)}$ is the density of $t+\mathrm{Exp}(\kappa)$,
of unit mass. If $\E_t^{\Q}[S_s]=S_t$ then \eqref{eq:ewma} averages a constant and
returns $S_t$. As $\kappa\to\infty$ the density converges weakly to the point mass
at $t$, and right continuity of $s\mapsto\E_t^{\Q}[S_s]$ gives $F_t\to S_t$.
\end{proof}

Shiller's design is self anchoring: it settles on a real flow and needs no
external price. The crypto design is basis anchored: it needs an observable
spot, and the peg is exact only in the fast funding limit or when the spot is a
$\Q$ martingale. This is the precise sense in which the two are the same
instrument with two benchmarks.

\section{The perpetual as a random time change}\label{sec:timechange}

Proposition~\ref{prop:pv} has a probabilistic reading that is useful for the
stochastic volatility analysis that follows. Take the peg form $a_s=b_s S_s$, so
that $a/b=S$; this is the crypto contract of Corollary~\ref{cor:cases}(ii) with a
possibly time varying intensity. Then \eqref{eq:pv} becomes
\begin{equation}\label{eq:peg-pv}
F_t=\E_t^{\Q}\!\left[\int_t^\infty b_s\,e^{-\int_t^s b_u\,du}\,S_s\,ds\right].
\end{equation}
The kernel $b_s\exp(-\int_t^s b_u\,du)$ is the density in $s$ of the first event
after $t$ of a Cox process with intensity $b$. Writing
$\tau=\inf\{s\ge t:\int_t^s b_u\,du\ge \mathcal E\}$ with $\mathcal E\sim
\mathrm{Exp}(1)$ independent of $\F$, the price is an expected spot at a random
time.

\begin{proposition}[Subordination]\label{prop:sub}
For the peg form, $F_t=\E_t^{\Q}[S_\tau]$, the expected spot at the first event
of a clock whose intensity is the funding rate.
\end{proposition}

\begin{proof}
Condition on the path. The law of $\tau$ given $\F$ has density
$b_s e^{-\int_t^s b_u du}$ on $[t,\infty)$, so
$\E_t[S_\tau]=\E_t\int_t^\infty b_s e^{-\int_t^s b_u du}S_s\,ds$, which is the
right side of \eqref{eq:peg-pv}.
\end{proof}

With constant $b=\kappa$ the clock is Poisson and $\tau=t+\mathrm{Exp}(\kappa)$,
recovering the exponentially weighted average of
Corollary~\ref{cor:cases}(ii). The peg is now transparent. If $S$ is a $\Q$
martingale and the clock is independent of $S$, optional sampling gives
$\E_t[S_\tau]=S_t$, so $F_t=S_t$ exactly, for any intensity. Deviations from the
peg are deviations from independence between the spot and the funding clock. The
stochastic volatility section identifies that dependence as the source of the
basis.

\begin{figure}[h]
\centering
\includegraphics[width=0.72\textwidth]{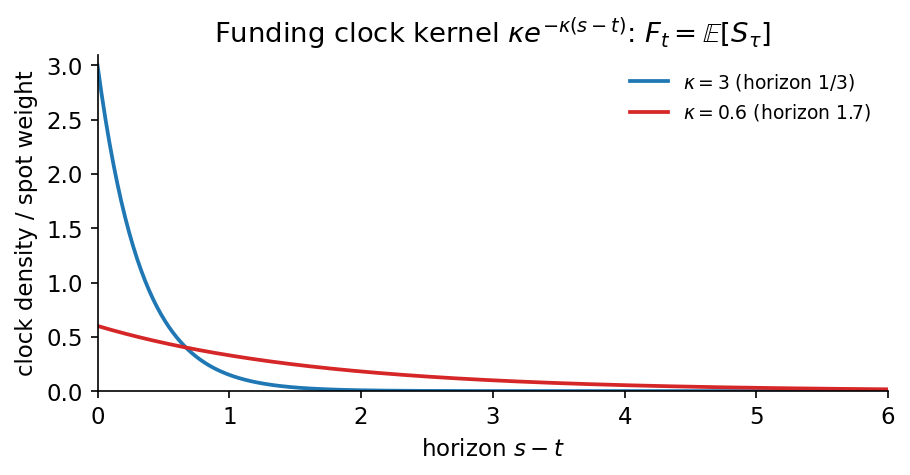}
\caption{The perpetual averages the future spot against the funding clock kernel
$\kappa e^{-\kappa(s-t)}$. A high funding intensity $\kappa$ concentrates the
weight near the present and tightens the peg; a low intensity samples further
out. The price is the expected spot at the first event of this clock.}
\end{figure}

This is a subordination in the sense of Bochner: the funding rule time changes the
spot by the random clock $\tau$, and the perpetual price is the value of the
subordinated process. When the intensity is constant the subordinator is the
exponential clock; when it depends on a volatility state the clock is doubly
stochastic, and the dependence between the subordinator and the spot is exactly
the channel through which volatility enters.

\section{A no arbitrage band under capped funding}\label{sec:band}

Real contracts cap funding, $|g_t|\le\bar g$, and the basis trade that would
enforce the peg is not free. Suppose an arbitrageur who is long spot and short
perpetual when $F_t>S_t$ pays a carry rate $\gamma>0$ for margin, collateral and
borrow, and collects funding $g_t=\min\{\kappa(F_t-S_t),\,\bar g\}$.

\begin{proposition}[Funding enforceable band]
\label{prop:band}
In the uncapped regime the basis is pinned at $F_t-S_t=\gamma/\kappa$. Funding
can discipline the basis only while $\kappa\,|F_t-S_t|\le\bar g$, that is inside
the band
\begin{equation}
|F_t-S_t|\le \frac{\bar g}{\kappa}.
\end{equation}
Outside this band funding saturates at $\bar g$ and the only remaining
discipline is a convergence trade, which requires a view on $S$.
\end{proposition}

\begin{proof}
The basis trade earns $g_t-\gamma$ per unit time. While uncapped,
$g_t=\kappa(F_t-S_t)$, and entry continues until $g_t=\gamma$, giving
$F_t-S_t=\gamma/\kappa$. Funding rises with the basis only until
$\kappa(F_t-S_t)=\bar g$; beyond that point $g_t\equiv\bar g$ regardless of the
basis, so a larger basis cannot be removed by collecting funding, and closing it
requires betting that $F$ will converge to $S$.
\end{proof}

The band $\bar g/\kappa$ is tight when $\kappa$ is large and the spot is hard.
For a private firm there is no hard $S$, so $\kappa$ is keyed to a reference
estimate and the band is centred on a soft number. Ackerer, Hugonnier and
Jermann (2024) give sharp deviation results in the frictional case.

\section{Replication and equivalence with a hard spot}\label{sec:equiv}

When a genuine tradable spot exists the choice of benchmark does not matter:
both designs return the spot.

\begin{proposition}[Equivalence with a hard spot]
\label{prop:equiv}
Let $S_t$ be a tradable spot paying dividend flow $D_t$, so under $\Q$,
$\E_t[dS_t]=(r_tS_t-D_t)\,dt$. Then the Shiller dividend settled perpetual
satisfies $F_t=S_t$. The interest adjusted basis funded perpetual satisfies
$F_t=S_t$ as well.
\end{proposition}

\begin{proof}
The Shiller perpetual has $\E_t[dF_t]=(r_tF_t-D_t)\,dt$ from
Section~\ref{sec:shiller}. Put $e_t=F_t-S_t$; subtracting the two drifts gives
$\E_t[de_t]=r_t e_t\,dt$, so $\beta_t e_t$ with $\beta_t=e^{-\int_0^t r_u du}$ is
a $\Q$ martingale. Under the transversality condition
$\lim_{T\to\infty}\E_t[\beta_T e_T]=0$ this forces $e_t=0$. For the basis funded
contract, choosing funding that nets the carry of the spot leaves
$a_t=r_tS_t-D_t+\kappa S_t$ and $b_t=\kappa+r_t$ effective, and
Proposition~\ref{prop:pv} returns the spot's own present value, namely $S_t$.
\end{proof}

So the difference between the designs is entirely about the case where there is
no hard spot. There, Shiller settles on a flow that does exist, while the crypto
contract must invent a reference.

\section{Stochastic volatility enters only through carry}\label{sec:sv}

Proposition~\ref{prop:sub} writes the peg as $F_t=\E_t^{\Q}[S_\tau]$, the expected
spot at the funding clock. This pins down the role of volatility exactly.

\begin{proposition}[Exact basis]\label{prop:exact}
Let $dS_t=c_t S_t\,dt+S_t\sqrt{v_t}\,dW_t$ under $\Q$ with adapted funding
intensity $b_t>0$. For the peg form,
\begin{equation}\label{eq:exact}
F_t-S_t=\E_t^{\Q}\!\left[\int_t^\infty e^{-\int_t^s b_u\,du}\,c_s S_s\,ds\right].
\end{equation}
\end{proposition}

\begin{proof}
Let $M_s=e^{-\int_t^s b_u du}$, continuous and of finite variation, so
$b_s M_s\,ds=-dM_s$ and $d(M_sS_s)=M_s\,dS_s+S_s\,dM_s$. By \eqref{eq:peg-pv} and
integration by parts, with $M_t=1$ and $M_\infty=0$,
\[
F_t=\E_t\!\int_t^\infty S_s(-dM_s)=\E_t\!\left[S_t+\int_t^\infty M_s\,dS_s\right].
\]
The integrand $M$ is bounded and predictable, so the stochastic integral against
the local martingale part of $S$ has zero mean, leaving
$\E_t\int_t^\infty M_s c_s S_s\,ds$.
\end{proof}

\begin{corollary}[Volatility irrelevance]\label{cor:svirr}
If $c\equiv0$, so $S$ is a $\Q$ martingale, then $F_t=S_t$ exactly, for every
funding intensity, including one driven by the same state as $S$ and under any
leverage $d\langle W,v\rangle\ne0$. Stochastic volatility alone does not move the
basis.
\end{corollary}

The corollary is optional sampling: in the filtration enlarged by the exponential
mark, $\tau$ is a stopping time and $S$ stays a martingale, so
$\E_t[S_\tau]=S_t$. Volatility reaches the basis only through the carry $c$. If a
volatility risk premium makes $c_t=c(v_t)$, then \eqref{eq:exact} is the funding
discounted expected carry, whose sign is the sign of the priced carry and whose
size is a covariance between the funding discount and that carry. This is the
perpetual analogue of the futures minus forward adjustment of Cox, Ingersoll and
Ross (1981): the gap is the covariance between the discount, here funding, and the
asset's drift, and it vanishes when the drift is zero or the funding is
deterministic. To trade the volatility itself one writes a perpetual on realized
variance (Section~\ref{sec:apps}); its benchmark flow is $v$, not a martingale
spot, so Corollary~\ref{cor:svirr} does not apply and it carries the exposure.

\section{Price discovery as filtering}\label{sec:filtering}

The pre-IPO problem is that the fundamental is latent. Let $x_t$ be the log
fundamental, unobserved, with
\begin{equation}
dx_t=\mu\,dt+\sigma\,dW_t,
\end{equation}
and let the market observe a noisy reference valuation
\begin{equation}\label{eq:obs}
dy_t=x_t\,dt+\eta\,dB_t,
\end{equation}
with $W,B$ independent. Write $p_t=\E[x_t\mid\F_t^{y}]$ for the market's log
price and $P_t$ for the posterior variance. Linear Gaussian filtering gives the
correction form
\begin{equation}
dp_t=\mu\,dt+\frac{P_t}{\eta^2}\,\big(dy_t-p_t\,dt\big),\qquad
\dot P_t=\sigma^2-\frac{P_t^2}{\eta^2},
\end{equation}
with steady state $P_\infty=\sigma\eta$ and gain $K_\infty=\sigma/\eta$. The
innovation $dy_t-p_t\,dt$ is the surprise in the reference, and the funding
correction $-\kappa(F_t-S_t)$ of Section~\ref{sec:unified} is exactly this term
in observer form, with $\kappa$ playing the role of the gain. The funding rule is
therefore a feedback observer, and the peg is its fixed point.

This reading makes the pre-IPO limitation precise. Before listing there is no
observation equation tied to an exogenous transaction price: the reference $y_t$
is itself an estimate, so $\eta$ is large, the steady state variance
$P_\infty=\sigma\eta$ is large, and the gain is small. The price is a wide
posterior that order flow can move, and there is no external term to pull it
back. At the Nasdaq open a hard observation arrives with $\eta\to0$, the gain
diverges, the posterior collapses onto the print, and the perpetual rebases.
Arbitrage closed is precisely the statement that an exogenous observation with
$\eta\to0$ is present. Before that, discovery is real, because order flow does
update $p_t$, but it is not closed, because nothing observes $x_t$ directly.

A mean reverting fundamental sharpens the picture. If instead
$dx_t=\alpha(\bar x-x_t)\,dt+\sigma\,dW_t$, the steady state posterior variance
solves $0=\sigma^2-2\alpha P_\infty-P_\infty^2/\eta^2$, giving
$P_\infty=\eta^2(\sqrt{\alpha^2+\sigma^2/\eta^2}-\alpha)$, which is increasing in
the observation noise $\eta$ and in the innovation variance $\sigma^2$ and
decreasing in the mean reversion $\alpha$. A firm with a stable, well anchored
fundamental is easy to discover even through a noisy reference, while a firm whose
value is itself a moving target stays uncertain until a hard observation arrives.
The gain $K_\infty=P_\infty/\eta^2$ is the funding intensity the market should run
to track it, which links the optimal $\kappa$ to the primitives.

\begin{figure}[h]
\centering
\includegraphics[width=0.74\textwidth]{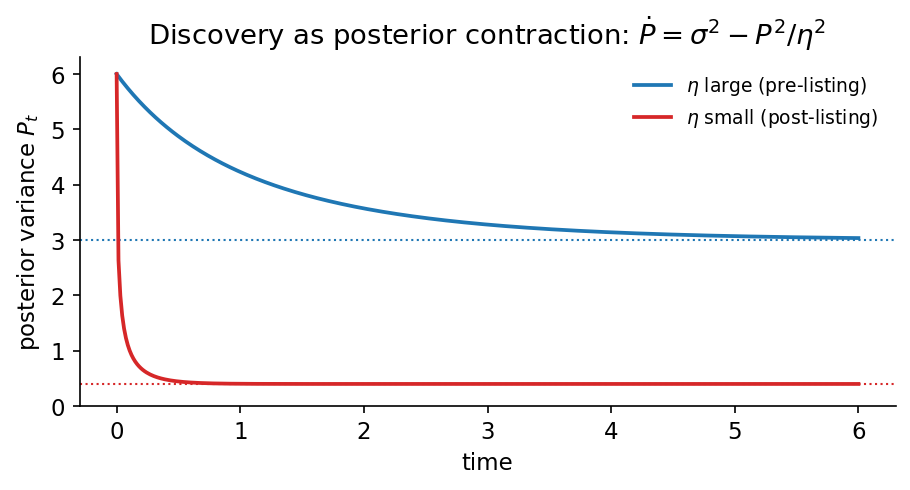}
\caption{The posterior variance $P_t$ solving $\dot P=\sigma^2-P^2/\eta^2$
contracts to its steady state $P_\infty=\sigma\eta$. A large observation noise
$\eta$, the pre-listing regime in which the only signal is a soft reference,
leaves a wide posterior; a small $\eta$, the hard observation supplied at the
Nasdaq open, collapses it. Discovery is the contraction of this variance.}
\end{figure}

\section{Bayesian martingale structure}\label{sec:bayes}

Three martingale facts organize the discovery story; each is standard, and
together they say what the perpetual is doing.

\paragraph{The price is a Doob martingale of the latent value.} Let
$\theta\in L^1(\F_\infty)$ be the value revealed at listing and
$p_t=\E[\theta\mid\F_t]$ the market's estimate. Then $(p_t)$ is a uniformly
integrable martingale, and by the martingale convergence theorem $p_t\to\theta$
almost surely and in $L^1$ as $\F_t\uparrow\F_\infty$. The conditional variance
$\mathrm{Var}(\theta\mid\F_t)$ is decreasing, and its drop is the information the
market has acquired; the listing is the time at which a hard observation makes
that drop discontinuous. This is the filtering of Section~\ref{sec:filtering}
stated as a convergence theorem: price discovery is the convergence of a Doob
martingale, and the funding rule is the mechanism that moves $p_t$ on each
innovation.

\paragraph{Cumulative settlement is a martingale difference.} Under the pricing
measure the per period settlement satisfies $\E_t^{\Q}[s_t+\Delta F_t]=0$ by
construction, so the cumulative gain $G_t=F_t+\Sigma_t$ is a $\Q$ martingale by
Proposition~\ref{prop:pv}. The perpetual admits no arbitrage if and only if such
an equivalent martingale measure exists, the fundamental theorem of asset pricing
localized to this contract (Harrison and Kreps 1979, Harrison and Pliska 1981).
The settlement stream is, up to sign, realized minus predicted, so a correctly
priced perpetual is a tradeable prequential score of the predictive sequence: it
pays when the prediction was wrong and nets to a martingale when the predictions
are calibrated.

\paragraph{Funding is a stochastic approximation.} Discretize and let order flow
move the price toward the target $m_k=\E[\theta\mid\F_k]$ in response to the
funding signal $g_k=\kappa(F_k-S_k)$,
\begin{equation}
F_{k+1}-F_k=-\alpha_k\,g_k+\xi_{k+1},\qquad \E[\xi_{k+1}\mid\F_k]=0.
\end{equation}
This is a Robbins and Monro (1951) recursion with fixed point the posterior mean.
With a decreasing gain satisfying $\sum_k\alpha_k=\infty$ and
$\sum_k\alpha_k^2<\infty$ it converges to $m_k$ when the target is fixed; with the
constant gain of a live contract it tracks a moving target, which is what a
nonstationary fundamental requires. The funding intensity $\kappa$ is the
learning rate, and the peg is the Bayesian fixed point. This is the same
correction as the Kalman gain of Section~\ref{sec:filtering}, now read as
stochastic optimization toward $\E[\theta\mid\F_t]$.

\section{A segmented market equilibrium for the basis}\label{sec:segment}

The premium decomposition of Section~\ref{sec:premium} treats segmentation as a
residual; here it is given structure. The perpetual is in zero net supply, a
swap, so every long is matched by a short. Suppose two classes of agent.
Unconstrained agents can hold the spot, or the IPO allocation, and the
perpetual, and arbitrage between them; constrained agents, retail or non-US, can
hold only the perpetual. Let $D\ge0$ be the constrained agents' aggregate desired
long position at zero basis, their excess demand for the exposure. Unconstrained
agents take the short side but require compensation, because shorting the
perpetual without a borrowable spot leaves unhedged risk and a margin cost; model
their inverse risk bearing capacity by $\lambda>0$, so they supply $q$ units short
only at basis $F-S=\lambda q$. Clearing $q=D$ gives the basis.

\begin{proposition}[Segmentation basis]\label{prop:seg}
With the funding cap of Proposition~\ref{prop:band},
\begin{equation}
F_t-S_t=\lambda\,D,\qquad
\lambda=\min\!\left(\frac{1}{\gamma_U},\ \frac{\bar g}{\kappa D}\right),
\end{equation}
where $\gamma_U$ is the aggregate risk tolerance of the unconstrained side. The
basis is the constrained demand priced at the cost of arbitrage capacity, capped
by the funding band.
\end{proposition}

A constant absolute risk aversion mean variance version delivers the same linear
form. Let the listing value be $\tilde v$ with $\E\tilde v=\bar v$ and
$\mathrm{Var}\,\tilde v=\sigma^2$. An agent of risk tolerance $\gamma$ holding $q$
units of the perpetual at basis $\beta=F-S$ has mean variance objective
$q(\bar v-S-\beta)-\tfrac{1}{2\gamma}q^2\sigma^2$, with optimum
$q^\star=\gamma(\bar v-S-\beta)/\sigma^2$. The unconstrained side faces an extra
margin and borrow penalty $\beta$ per unit short, so its supply of short units is
$\gamma_U\beta/\sigma^2$, increasing in the basis. The constrained side cannot
short the spot and brings inelastic excess demand $D$. Clearing the swap,
$\gamma_U\beta/\sigma^2=D$, gives $\beta=\sigma^2 D/\gamma_U$, the linear form
with $\lambda=\sigma^2/\gamma_U$, valid until the implied funding $\kappa\beta$
reaches the cap $\bar g$, beyond which Proposition~\ref{prop:band} takes over. As
$\gamma_U\to\infty$ the basis vanishes and the peg is exact, the deep liquid spot
of Section~\ref{sec:equiv}. As arbitrage capacity falls, through scarce borrow,
capital limits, or the funding cap, the basis widens in proportion to constrained
demand. For a private firm there is no borrow at all and the IPO allocation is
scarce, so $\gamma_U$ is small and $D$ is large, and the basis is large and
positive. The premium is then structural, a price of access and risk bearing, not
a forecast error, which is why it can persist without being arbitraged and why it
resolves only when listing and the end of the lockup restore $\gamma_U$.

\section{Estimation by generative Bayesian computation}\label{sec:gbc}

The model that has accumulated, a latent fundamental with stochastic volatility
observed through a capped funding mechanism, is a nonlinear non Gaussian state
space model, and the Gaussian filter of Section~\ref{sec:filtering} is only its
linear special case. Inference on the latent path and the parameters
$\Theta=(\mu,\sigma,\eta,\kappa,c,\bar g)$ from the observed series of price,
funding and volume is therefore likelihood free in practice. Generative Bayesian
computation (Polson and Sokolov 2024) replaces the likelihood with simulation and
a learned posterior map:
\begin{enumerate}[leftmargin=1.4em,itemsep=1pt]
\item draw $\Theta^{(i)}$ from the prior, simulate a path of price, funding and
volume, and reduce each to a vector of summary statistics $\mathbf s^{(i)}$;
\item train a conditional density or quantile estimator $\mathbf s\mapsto\Theta$
on the pairs $(\mathbf s^{(i)},\Theta^{(i)})$;
\item evaluate the estimator at the observed $\mathbf s^{\star}$ to draw from the
posterior $\Theta\mid\mathbf s^{\star}$;
\item propagate to the filtered fundamental $p_t$ and the basis decomposition of
Section~\ref{sec:premium}, with posterior uncertainty.
\end{enumerate}
The informative summaries are the mean and persistence of the funding rate, which
identify $\kappa$ and the band, the realized volatility and the level of the
funding carry, which identify $\sigma$ and $c$ of Section~\ref{sec:sv}, and the
basis level and its mean reversion, which identify the segmentation term. The
output is a posterior over the carry basis and the premium components rather than
point estimates.

\section{Discrete funding mechanics}\label{sec:mechanics}

The continuous settlement \eqref{eq:peg-pv} is an idealization of a discrete rule.
A live contract marks to a mark price $M_t$ and references an index price $I_t$,
and at funding times spaced by $h$, typically eight hours, the long pays the short
\begin{equation}
g_t=\Big(c+ \mathrm{clip}\big(\tfrac{M_t-I_t}{I_t},\,-\bar g,\,\bar g\big)\Big)\,h\,N_t,
\end{equation}
where $N_t$ is the position notional, $c$ a fixed interest component, and the clip
imposes the cap of Proposition~\ref{prop:band}. The premium $(M_t-I_t)/I_t$ is the
discrete basis, and the per period funding is an affine, clipped function of it, so
the continuous intensity $\kappa$ is the slope of the funding rule and $\bar g$ its
saturation. Two features matter for the readings above. First, funding is charged
on notional and settled in the quote asset, here USDC, so a long is exposed to the
funding rate as a stochastic carry, which is the $b_t$ of
Section~\ref{sec:timechange}; by Corollary~\ref{cor:svirr} this carry moves the
basis only when it fails to match the spot's own drift. Second, between funding
times the contract is disciplined only by the mark mechanism and liquidations, not
by the funding flow, so the band of Proposition~\ref{prop:band} is a statement
about averages across funding intervals rather than an instantaneous bound. For a
private underlying the index $I_t$ is not a transaction price but a reference
valuation, which is the soft anchor that Sections~\ref{sec:filtering}
and~\ref{sec:segment} analyze.

\section{The SpaceX pre-IPO market: a worked example}\label{sec:spacex}

SpaceX was private with no continuously observable spot and no dividend, so it is
the Shiller use case stated literally: an equity claim whose price is difficult
or impossible to observe before listing. The instrument the market built for it
was the crypto funding variant of Corollary~\ref{cor:cases}(ii), not the
dividend variant. The contract was USDC settled, traded continuously with no
expiry, referenced an implied pre-IPO valuation, and transitioned automatically
to the live equity perpetual at listing with no rollover. Holders received price
exposure only, with no voting rights, no dividends, and no claim on assets.

\begin{table}[h]
\centering
\small
\begin{tabular}{lll}
\hline
Date & Event & Price \\
\hline
May 2026 & Pre-IPO perpetual high & above \$220 \\
Early June & Pre-IPO perpetual & near \$180 \\
June 11 & IPO priced & \$135 \\
June 12 & Nasdaq open, then first close & \$150, then \$160.95 \\
June 12 & SPCX-USDC perpetual & near \$176 \\
June 16 & All time high & \$225.64 \\
June 18 & Close & \$185.00 \\
\hline
\end{tabular}
\caption{SpaceX listing and the SPCX perpetual, May to June 2026.}
\end{table}

The pre-IPO market was large. Binance listed SPCXUSDT on May 21, 2026, the
contract became its second most traded futures behind BTCUSDT, and cumulative
SpaceX linked derivative volume ran above \$9 billion across venues with single
day volume near \$5.7 billion. On listing day the SPCX-USDC perpetual on
Hyperliquid traded near \$176, about 30 percent above the offer, with 24 hour
volume above \$233 million and open interest above \$263 million.

\paragraph{The inversion.} The crypto consensus near \$180 was a better forecast
of the secondary clearing price near \$185 than the underwritten \$135 was. The
fixed offer was an administered price set below the market, the standard
underpricing of a mega-IPO, and the perpetual called it. So the optional stopping
identity $\E^{\Q}[F_T]=F_0$ must be read against the secondary tape and not the
offer: the terminal value $S_T$ is the listing price the market discovered, near
\$185, not the \$135 print, which carried no order flow. With $F_0\approx 180$ and
$S_T\approx 185$ the perpetual was a $\Q$ martingale to within a few percent. The
pre-IPO long near \$180 was not underwater against the tape.

\begin{figure}[h]
\centering
\includegraphics[width=0.78\textwidth]{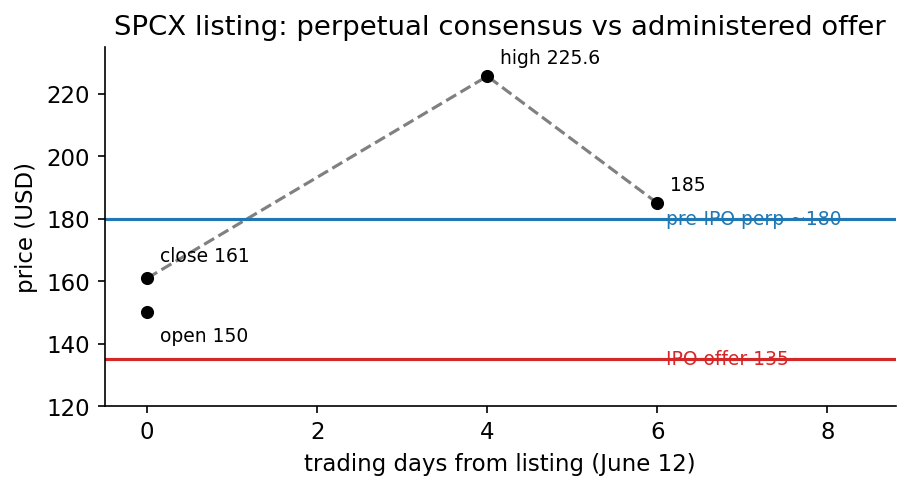}
\caption{The pre-IPO perpetual consensus near \$180 forecast the secondary
clearing price more accurately than the administered \$135 offer. Markers are
observed SPCX prints; the dashed line is a guide. The offer sat well below the
market the perpetual had already discovered.}
\end{figure}

\paragraph{Filtering reading.} The posterior of Section~\ref{sec:filtering} did
not collapse at the offer. It collapsed over June 12 to 16 as Nasdaq trading
delivered the first hard observations, with the posterior mean migrating
$135\to150\to161\to225\to185$. The posterior variance stayed large, a four day
range of \$149 to \$225, because the lockup keeps the float tiny, so $\eta$ in
\eqref{eq:obs} is still large and the steady state $P_\infty=\sigma\eta$ has not
contracted. The December lockup expiry enlarges the float, lowers $\eta$, raises
the gain, and is the point at which the posterior should finally settle toward
fundamental value.

\section{Reading the premium}\label{sec:premium}

The pricing sections above are stated under $\Q$. The pre-IPO \$180 against the
\$135 offer is a statement under the physical measure $\Prob$ with frictions, and
it decomposes as
\begin{equation}
\log\!\frac{F_t}{S_t^{\,\mathrm{IPO}}}
=\underbrace{\theta_t}_{\text{risk premium}}
+\underbrace{\zeta_t}_{\text{segmentation}}
+\underbrace{\delta_t}_{\text{discovery}}
+\underbrace{\beta_t^{-1}\,b_t}_{\text{bubble}} .
\end{equation}
The risk premium $\theta_t$ is the change of measure carried by a leveraged
holder of price risk and has ambiguous sign. The segmentation term $\zeta_t$ is
the value of synthetic access of Section~\ref{sec:segment}, $\zeta_t\approx
\lambda D/S_t$, and is strictly positive. The discovery term $\delta_t$ is the
gap between the administered offer and the market clearing value, also positive
here. The bubble term has $\beta_t B_t$ a $\Q$ martingale and is admissible only
if the transversality condition of Proposition~\ref{prop:pv} fails; its
behavioural counterpart is extrapolative demand. The realized outcome, a
secondary price near \$185 against the \$135 offer, shows that the discovery and
segmentation terms dominated and that the bubble term was small. The perpetual
was not euphoric. It was right. The December lockup expiry, by enlarging the
float and relaxing $\zeta_t$, is the natural experiment that isolates the
remaining terms.

The magnitudes are consistent with the segmentation model. The log premium of the
perpetual over the secondary price was small, the perpetual near \$176 to \$180
against a tape near \$185, so $\zeta_t+\delta_t$ measured against the secondary
market is close to zero and the bulk of the \$180 over \$135 gap is the discovery
term $\delta_t$, the underwriters' discount. Against the offer the implied premium
$\log(180/135)\approx0.29$ splits into a discovery term of about $\log(185/135)
\approx0.31$ and a small negative residual, leaving little room for a bubble. By
Proposition~\ref{prop:seg} the residual segmentation premium pins down the ratio
$\sigma^2 D/\gamma_U$ of constrained demand to arbitrage capacity; its near
vanishing once the spot is hard is the model's prediction that the basis collapses
when borrow and allocation become available at listing.

\section{Other applications to equities}\label{sec:apps}

Beyond pre-IPO discovery the same instrument addresses several long standing
problems.

\paragraph{Private and pre-IPO price discovery.} The SpaceX mechanism applies to
any late stage private firm, such as OpenAI, Stripe or Databricks. The perpetual
consensus is a continuous when-issued price for a claim that otherwise reprices
only at funding rounds, and Section~\ref{sec:filtering} describes exactly how
informative it is.

\paragraph{Synthetic shorting and the implied borrow rate.} Where stock borrow is
expensive or unavailable, the perpetual provides synthetic short exposure, and
the funding rate clears the short demand. In the steady state of
Proposition~\ref{prop:band} the funding rate equals the marginal cost of the
arbitrage, so it is a continuous, observable proxy for the securities lending fee,
cleaner than the opaque stock loan market.

\paragraph{Restricted and locked-up equity.} An employee or pre-IPO holder who
cannot yet sell can short the perpetual to hedge, paying the funding carry as the
cost of the hedge. This is the concentrated position problem, and it is the live
question for SpaceX holders through the December lockup.

\paragraph{Untraded macro and real indices.} Shiller's original targets, the
Case-Shiller home price index, the consumer price index, wage and GDP indices,
all have an observable flow and so suit the dividend variant of
Corollary~\ref{cor:cases}(i). A perpetual on such an index lets a household hedge
human capital or housing risk without a transactable spot.

\paragraph{Perpetual dividend strips.} A perpetual settled on a dividend index is
a traded Gordon model: by Corollary~\ref{cor:cases}(i) with a constant flow its
price is $D/\kappa$, separating dividend risk from price risk with no expiry and
no roll.

\paragraph{Roll-free index, factor and variance exposure.} Funding replaces the
roll, removing the negative roll yield of calendar futures. A perpetual on
realized variance, with funding keyed to the realized variance of the underlying,
gives no expiry volatility exposure and is the direct instrument for the
stochastic volatility of Section~\ref{sec:sv}.

\paragraph{Event and when-issued discovery.} The same continuous market prices
SPACs, spinoffs, direct listings and merger targets, where the administered or
infrequent price leaves a discovery gap of the kind seen at the SpaceX listing.

\paragraph{Collateral and continuous access.} Because the contract is cash settled
in a stablecoin and trades without interruption, it gives round the clock equity
exposure outside exchange hours and a collateral efficient synthetic position for
holders who want exposure without custody of the share. The same property that
makes it useful, settlement against a reference rather than delivery of the asset,
is the source of the soft anchor of Section~\ref{sec:segment}: the instrument is
only ever as well disciplined as the reference it settles against, which is the
recurring theme of this paper.

\section{Discussion}\label{sec:discussion}

\textbf{Establishes.} A perpetual contract can manufacture a tradable, continuous
price for a private equity claim that has no spot and no dividend, the original
Shiller goal carried out for equity rather than for homes or wages. The funding
peg sustained a deep two sided market, and its pre-IPO consensus forecast the
secondary clearing price more accurately than the bookbuilt offer.

\textbf{Does not establish.} By Proposition~\ref{prop:band} the peg holds only
inside the band $\bar g/\kappa$, and by Section~\ref{sec:filtering} the pre-IPO
band is centred on a soft reference rather than a hard spot, so before listing the
discipline is to an estimate and not to an arbitrage enforceable price. By
Proposition~\ref{prop:seg} a large part of the premium is the segmentation term
rather than information. The forecast was vindicated against the secondary tape,
but the float is still small and the posterior variance still large, so the test
of fundamental value as opposed to clearing price is deferred to the December
lockup expiry.

\textbf{Reading.} Shiller's dividend settled design prices an unobservable asset
by anchoring to a real flow and is self contained. The crypto funding design
anchors to a basis and needs an external reference, which for a private firm is
soft. By Proposition~\ref{prop:equiv} the two coincide once a hard spot exists,
and by Proposition~\ref{prop:sub} both are the expected spot at a funding clock.
The robust part is the instrument design. The fragile part, made precise by the
filtering and segmentation readings, is the quality of the anchor when no hard
price exists. The dividend variant that Shiller proposed, settling on a real
cash flow, remains largely unbuilt, and is the natural next instrument for the
macro and real indices of Section~\ref{sec:apps}.

\textbf{Predictions.} The framework is falsifiable. First, by
Proposition~\ref{prop:seg} the basis should collapse toward zero as the float
grows, so the perpetual to spot premium should fall sharply across the December
lockup expiry; a persistent premium afterward would reject the segmentation
reading in favour of a bubble. Second, by Section~\ref{sec:apps} the funding rate
should track the securities lending fee once the stock is borrowable, so the two
should converge post listing. Third, by Corollary~\ref{cor:svirr} the basis is
insensitive to the spot's volatility and responds only to its carry, so a near
martingale equity should peg tightly however volatile it is, while a variance
perpetual, whose benchmark is volatility itself, should carry any volatility risk
premium.
Each is measurable with the data the contracts already generate, and
Section~\ref{sec:gbc} gives the estimation route.

\textbf{Conclusion.} Shiller's perpetual future was a proposal to price what
cannot be traded. The crypto markets supplied the funding mechanism that made it
work at scale, and the SpaceX listing showed that the resulting price can beat the
administered alternative for a private equity claim. The unifying object is the
funding rate, which is simultaneously a discount rate in
Proposition~\ref{prop:pv}, a clock in Proposition~\ref{prop:sub}, an observer gain
in Section~\ref{sec:filtering}, and a learning rate in Section~\ref{sec:bayes}.
What remains fragile is the anchor, and the cleanest open instrument is the
dividend settled variant Shiller first described, which needs no anchor at all.

\section*{References}

\begin{itemize}[leftmargin=1.4em,itemsep=2pt]
\item Ackerer, D., Hugonnier, J., and Jermann, U. (2024). Perpetual Futures
Pricing. Working paper, EPFL and Wharton.
\item Cox, J. C., Ingersoll, J. E., and Ross, S. A. (1981). The relation between
forward prices and futures prices. Journal of Financial Economics, 9(4),
321--346.
\item Doob, J. L. (1953). Stochastic Processes. Wiley, New York.
\item Harrison, J. M., and Kreps, D. M. (1979). Martingales and arbitrage in
multiperiod securities markets. Journal of Economic Theory, 20, 381--408.
\item Harrison, J. M., and Pliska, S. R. (1981). Martingales and stochastic
integrals in the theory of continuous trading. Stochastic Processes and their
Applications, 11, 215--260.
\item He, S., Manela, A., Ross, O., and von Wachter, V. (2022). Fundamentals of
Perpetual Futures. arXiv:2212.06888.
\item Kalman, R. E. (1960). A new approach to linear filtering and prediction
problems. Journal of Basic Engineering, 82(1), 35--45.
\item Liptser, R. S., and Shiryaev, A. N. (2001). Statistics of Random Processes.
Springer, Berlin.
\item Polson, N. G., and Sokolov, V. (2024). Generative Bayesian Computation.
Working paper.
\item Robbins, H., and Monro, S. (1951). A stochastic approximation method.
Annals of Mathematical Statistics, 22(3), 400--407.
\item Shiller, R. J. (1993a). Measuring Asset Values for Cash Settlement in
Derivative Markets: Hedonic Repeated Measures Indices and Perpetual Futures.
Journal of Finance, 48(3), 911--931.
\item Shiller, R. J. (1993b). Macro Markets: Creating Institutions for Managing
Society's Largest Economic Risks. Oxford University Press, Oxford.
\item Market and venue data: cnbc.com, investing.com, tradingview.com,
cryptobriefing.com and barchart.com, June 2026.
\end{itemize}

\end{document}